\documentclass[11pt]{article}
\usepackage[utf8]{inputenc}
\usepackage{amsmath}
\usepackage{amssymb}
\usepackage{amsthm}
\usepackage{mathtools}
\usepackage{hyperref}
\usepackage{booktabs}
\usepackage{algorithm,algpseudocode}
\usepackage{geometry}
\usepackage{booktabs}
\usepackage{authblk}

\newcommand{\F}{\mathbb{F}}

\newtheorem{theorem}{Theorem}
\newtheorem{corollary}{Corollary}
\newtheorem{conjecture}{Conjecture}
\renewcommand\footnotemark{}
\title{Classification of Planar Monomials Over Finite Fields of Small Order\thanks{This  work  was  funded  by Deutsche  Forschungsgemeinschaft  (DFG) under Germany's Excellence Strategy - EXC 2092 CASA - 390781972.}}
\author[1]{Christof Beierle}
\affil[1]{Faculty of Computer Science, Ruhr University Bochum, Bochum, Germany}
\author[2]{Patrick Felke}
\affil[2]{University of Applied Sciences Emden-Leer, Emden, Germany}
\date{\vspace{-5ex}}

\begin{document}

\maketitle

\begin{abstract}
    For all finite fields of order up to $2^{30}$, we computationally prove that there are no planar monomials besides the ones already known.
\end{abstract}

\section{Introduction}
Let $p$ be a prime and let $n$ be a positive integer. A polynomial $f \in \F_{p^n}[X]$ is called \emph{planar} if, for all $a \in \F_{p^n}\setminus \{0\}$, the polynomial $f(X+a)-f(X)$ is a permutation polynomial over $\F_{p^n}$. In this work, we focus on monomials $X^k$ for a positive integer $k$. If $f \in \F_{p^n}[X]$ is planar and of the form $X^k$, it is called a \emph{planar monomial} over $\F_{p^n}$. It is easy to see that $X^k$ is a planar monomial over $\F_{p^n}$ if and only if \[\Delta_k(X)\coloneqq (X+1)^k - X^k\] is a permutation polynomial over $\F_{p^n}$. Note that over finite fields of characteristic 2, planar polynomials do not exist. This is because solutions of the equation $f(x+a) - f(x) = b$ come in pairs $(x,x+a)$. In the following, we therefore only consider the case of $p$ being odd.

When it comes to classifying planar polynomials, an important concept is the notion of \emph{graph equivalence} (better known as CCZ-equivalence for $p=2$) ~\cite{DBLP:journals/dcc/CarletCZ98}. Two polynomials $f,g \in \F_{p^n}[X]$ are called \emph{graph equivalent} if there exists an invertible affine transformation $\phi$ over $\F_{p^n} \times \F_{p^n}$ such that $\Gamma_g = \phi(\Gamma_f)$, where $\Gamma_h \coloneqq \{(x,h(x)) \mid x \in \F_{p^n}  \}$ denotes the \emph{graph} of  $h \in \F_{p^n}[X]$. It is well known that graph equivalence preserves the planarity property of polynomials (see~\cite{DBLP:conf/waifi/KyureghyanP08}).

The following theorem lists the known planar monomials over finite fields, see e.g.,~\cite{DBLP:journals/dcc/Zieve15}. 
 
\begin{theorem}\label{thm:known}
If $0 \leq i < n$ and $p \frac{n}{\gcd(i,n)}$ is odd, then $X^{p^i+1}$ is a planar monomial over $\F_{p^n}$. Further, if $p=3$ and $2 < i < n$ and $\gcd(i,2n)=1$, the monomial $X^{e}$ with $e = \frac{3^i+1}{2}$ is planar over $\F_{p^n}$.
\end{theorem}

The planarity of monomials $X^e$ with $e = \frac{3^i+1}{2}$ over finite fields of characteristic $3$ was discovered in~\cite{DBLP:journals/dcc/CoulterM97} and a special case of it independently in~\cite{DBLP:journals/aaecc/HellesethS97}. Before that, a conjecture by Dembowski and Ostrom (see~\cite{dembowski1968planes}) stated that a planar polynomial in $\F_{p^n}[X]$ must necessarily be of the form $\sum_{i,j \in \{0,\dots,n-1\}} \alpha_{i,j}X^{p^i+p^j} + \sum_{i \in \{0,\dots,n-1\}} \beta_i X^{p^i} + c$. Note that the terms $\beta_i X^{p^i}$ and the constant term $c$ could be omitted if we are only interested in polynomials up to graph equivalence. Although this conjecture is false for $p=3$ (because of the construction in~\cite{DBLP:journals/dcc/CoulterM97} and~\cite{DBLP:journals/aaecc/HellesethS97}), it is completely open for $p \geq 5$. Besides over the finite fields $\F_{p^i}$ for $i \in \{1,2,3,4\}$, a classification of planar monomials is not known. The following conjecture is a longstanding open problem. 
\begin{conjecture}\label{conj:open}
Up to graph equivalence, the monomials described in Theorem~\ref{thm:known} are the only planar monomials. 
\end{conjecture}
In~\cite{DBLP:journals/dcc/Zieve15}, Zieve verified this conjecture for finite fields of sufficiently large order. In particular, Zieve proved that for $p^n \geq (k-1)^4$ and $p \nmid k$, the monomial $X^k \in \F_{p^n}[X]$ is planar only if it is of the form as in Theorem~\ref{thm:known}. We are not aware of any other reference that states a systematic check of Conjecture~\ref{conj:open} for finite fields of small order. In this work, we provide such a result by computationally verifying the conjecture for all finite fields of order at most $2^{30}$, as well as for the finite fields $\F_{37^6},\F_{17^8},\F_{19^8}$, $\F_{11^9}$, and $\F_{11^{10}}$. In particular, there are no other planar monomials in those finite fields besides the ones already known.

\paragraph{Related work.}  Helleseth, Rong and Sandberg conducted extensive computer search in the 1990s to classify $\ell$-uniform monomials. We recall that a monomial $X^k \in \F_{p^n}[X]$ is called $\ell$-\emph{uniform} if $\max_{b \in \F_{p^n}} \lvert \{x \in \F_{p^n} \mid (x+1)^k -x^k = b \} \rvert$ is equal to $\ell$.  These numerical results are well known as the Helleseth-Rong-Sandberg tables (see e.g.,~\cite{DBLP:journals/tit/HellesethRS99}). This search did not give any new insights to Conjecture \ref{conj:open}. In the case of $p=2$, Dobbertin and Canteaut computationally classified all almost perfect nonlinear (APN) monomials, i.e., $2$-uniform monomials, up to $n=26$. Moreover, in~\cite{leander2008exponents} Leander and Langevin classified all almost bent (AB) monomials up to $n=33$. Also for $p=2$, Edel classified all APN monomials up to $n \leq 34$ and for $n \in \{36,38,40,42\}$ (see~\cite[p.\@ 422]{carlet_2021}).

\section{Method}
We will make extensive use of the following characterization of graph equivalence for monomials.
\begin{theorem}[\cite{DBLP:journals/dcc/Dempwolff18}]
The monomials $X^k$ and $X^\ell$ over $\F_{p^n}$ are  graph equivalent if and only if there exists an integer $a$ with $0 \leq a < n$ such that $\ell = p^ak \mod p^n-1$ or $k\ell = p^a \mod p^n-1$.
\end{theorem}

This characterization simplifies for the case of planar monomials. In particular, a monomial $X^k$ being planar over $\F_{p^n}$ implies that $\gcd(k,p^n-1) = 2$, see~\cite{DBLP:journals/dcc/CoulterM97}. We therefore have the following.

\begin{corollary}
\label{cor:equivalence}
Two planar monomials $X^k$ and $X^\ell$ over $\F_{p^n}$ are graph equivalent if and only if there exists an integer $a$ with $0 \leq a < n$ such that $\ell = p^ak \mod p^n-1$.
\end{corollary}
\begin{proof}
Suppose there exists an integer $a \in \{0,\dots,n-1\}$ such that $k \ell = p^a \mod p^n-1$. We then have $p^{n-a}k \cdot \ell = 1 \mod p^n-1$,  which implies that $\gcd(\ell,p^n-1) = 1$, a contradiction to the planarity of $X^{\ell}$.
\end{proof}

For an odd prime $p$, the classification of planar monomials over the fields $\F_{p},\F_{p^2},\F_{p^3}$, and $\F_{p^4}$ is known from the works~\cite{johnson1987projective},~\cite{coulter2006classification},~\cite{coulter_new}, and~\cite{DBLP:journals/ffa/CoulterL12}, respectively.

\begin{theorem}[\cite{johnson1987projective,coulter2006classification,coulter_new,DBLP:journals/ffa/CoulterL12}]
\label{thm:planar_subfield}
Let $k$ be a positive integer and let $p$ be an odd prime. Then, the following assertions hold:
\begin{enumerate}
    \item The monomial $X^k$ is planar over $\F_p$ if and only if $k = 2 \mod (p-1)$.
    \item The monomial $X^k$ is planar over $\F_{p^2}$ if and only if $k = 2 \mod (p^2-1)$ or $k = 2p \mod (p^2-1)$.
    \item The monomial $X^k$ is planar over $\F_{p^3}$ if and only if $k = p^i + p^j \mod (p^3-1)$ for $i,j \in \{0,1,2\}$.
    \item If $p \geq 5$, the monomial $X^k$ is planar over $\F_{p^4}$ if and only if $k =2p^j \mod (p^4-1)$ for $j \in \{0,1,2,3\}$.
\end{enumerate}
\end{theorem}

Our method for finding all planar monomials $X^k$ over $\F_{p^n}$ for $k=2,\dots,p^n-2$ is to first restrict to all exponents $k$ which fulfill $\gcd(k,p^n-1) = 2$ and for which $X^k$ is planar over the subfields $\F_{p},\F_{p^2},\F_{p^3}$, and $\F_{p^4}$ (if they exist) based on the conditions of Theorem~\ref{thm:planar_subfield}. If $n$ has a non-trivial divisor $i$ with $i>4$ and if we have already completed the classification of planar monomials over $\F_{p^i}$, we can also enforce the planarity of $X^k$ over $\F_{p^i}$ as a necessary condition. For each such monomial $X^k$ that is left, we only check one member of its equivalence class with respect to the equivalence relation stated in Corollary~\ref{cor:equivalence}. More precisely, we only check the exponent $k \in \{2,\dots,p^n-2\}$ if $k$ is the smallest integer in the set $\{p^ak \mod p^n-1 \mid a \in \{0,\dots,n-1\}\}$. 

\renewcommand{\algorithmicrequire}{\textbf{Input:}}
\renewcommand{\algorithmicensure}{\textbf{Output:}}
\begin{algorithm}
	\caption{\textsc{PlanarityCheck}} \label{alg:check}
	\begin{algorithmic}[1]
		\Require Positive integer $k$
		\Ensure 1 if $X^k$ is a candidate for being planar over $\F_{p^n}$ and 0 if $X^k$ is not planar over $\F_{p^n}$. 
		    \vspace{.1em}
		    \State initialize $Q = \{\}$ and $E = \{\}$
		    \For{$i=0$ to $N$}
		    \State sample $x \gets \F_{p^n}$ uniformly at random
		    \If{$x$ not in $Q$}
		    \If{$(x+1)^k - x^k$ in $E$}
		    \State \Return 0
		    \EndIf
		    \State insert $x$ in $Q$
		    \State insert $(x+1)^k - x^k$ in $E$
		    \EndIf
		    \EndFor
		    \State \Return 1
	\end{algorithmic}
\end{algorithm}
For each remaining exponent $k$, for checking the planarity of the monomial $X^k$ over $\F_{p^n}$ we make use of the assumption that $x \mapsto \Delta_k(x)$ behaves like a random mapping for almost all $k$ and make use of the birthday paradox. In other words we conduct a collision search on the image of the mapping $x \mapsto (x+1)^k- x^k$ and expect to find a collision after $\sqrt{p^n}$ trials. The method is formally described in Algorithm~\ref{alg:check}. Note that when Algorithm~\ref{alg:check} returns 1 on input $k$, the result might be a false positive and one further needs to evaluate the planarity of $X^k$. In order to avoid false positive candidates for $k$ with high probability, we set the parameter $N$ to $20\sqrt{p^n}$. This is justified as follows. Let $\text{Prob}_{p^n,N}$ denote the probability that there is no collision after $N$ steps. By our assumption on $\Delta_k(X)$, this probability is the same as the probability that there is no collision for choosing uniformly at random $N$ times in a set of $p^n$ elements. This probability is upper bounded by $e^{\frac{-N(N-1)}{p^n}}$ (see e.g., \cite{brink}), which gives in our case \[\text{Prob}_{p^n,20\sqrt{p^n}}\leq e^{\frac{-20\sqrt{p^n}\left(20\sqrt{p^n}-1\right)}{p^n}}\approx e^{-400}.\]
Note that the speed-up by lowering $N$ does not compensate for the enhanced probability of having to deal with false positives. 

We have implemented this search in c++ using the NTL library~\cite{shoup2001ntl} and our implementation is available in~\cite{code}.

\begin{table}[ht]
\caption{Number of graph equivalence classes containing planar monomials over $\F_{p^n}$ for $n\geq 5$.  \label{tab:results}}
    \centering
\begin{tabular}{cccccccccccccccccc}
\toprule 
$(n,p)$ & 3 & 5 & 7 & 11 & 13 & 17 & 19 & 23 & 29 & 31 & 37 & 41 & 43 & 47 & 53 & 59 & 61 \\ 
\midrule 
5 & 4 & 3 & 3 & 3 & 3 & 3 & 3 & 3 & 3 & 3 & 3 & 3 & 3 & 3 & 3 & 3 & 3 \\ 
6 & 3 & 2 & 2 & 2 & 2 & 2 & 2 & 2 & 2 & 2 & 2 &  &  &  &  &  &    \\ 
7 & 6 & 4 & 4 & 4 & 4 & 4 & 4 &  &  &  &  &  &  &  &  &  &    \\ 
8 & 4 & 1 & 1 & 1 & 1 & 1 & 1 &  &  &  &  &  &  &  &  &  &    \\ 
9 & 7 & 5 & 5 & 5 &  &  &  &  &  &  &  &  &  &  &  &  &    \\ 
10 & 6 & 3 & 3 & 3 &  &  &  &  &  &  &  &  &  &  &  &  &    \\ 
11 & 10 & 6 &  &  &  &  &  &  &  &  &  &  &  &  &  &  &   \\ 
12 & 5 & 2 &  &  &  &  &  &  &  &  &  &  &  &  &  &  &   \\ 
13 & 12 &  &  &  &  &  &  &  &  &  &  &  &  &  &  &  &   \\ 
14 & 9 &  &  &  &  &  &  &  &  &  &  &  &  &  &  &  &   \\ 
15 & 11 &  &  &  &  &  &  &  &  &  &  &  &  &  &  &  &   \\ 
16 & 8 &  &  &  &  &  &  &  &  &  &  &  &  &  &  &  &   \\ 
17 & 16 &  &  &  &  &  &  &  &  &  &  &  &  &  &  &  &   \\ 
18 & 10 &  &  &  &  &  &  &  &  &  &  &  &  &  &  &  &  \\ 
\bottomrule
\end{tabular} 
\end{table}

\section{Results}
We ran our algorithm on two AMD EPYC 7742 server processors and conducted the search for all finite fields of order at most $2^{30}$, as well as for the finite fields $\F_{37^6},\F_{17^8},\F_{19^8}$, $\F_{11^9}$, and $\F_{11^{10}}$. Note that we did not perform any search over $\F_{p^i}$ for $i \leq 4$ since all planar monomials are classified in those cases. The computationally hardest case took a few days to finish. For all finite fields for which we finished the computation, Table~\ref{tab:results} lists the number graph equivalence classes containing planar monomials. Besides the examples listed in Theorem~\ref{thm:known}, no new planar monomial was found.

\subsection*{Acknowledgment}
We thank Frederik Gosewehr for double-checking the correctness of our implementation.

\end{document}